\documentclass[12pt]{article} 
\usepackage{amsfonts, amsmath, amssymb, amsthm, xcolor} 
\usepackage{tikz-cd}
\usepackage[utf8]{inputenc}

\swapnumbers
\newtheorem{theorem}{Theorem}[section] 
\newtheorem{definition}[theorem]{Definition} 
\newtheorem{lemma}[theorem]{Lemma} 
 
\newtheorem{remark}[theorem]{Remark} 
\newtheorem{example}[theorem]{Example} 
\newtheorem{proposition}[theorem]{Proposition}

\newcommand{\1}{{\mathsf 1}}
\newcommand{\A}{{\mathcal A}}
\newcommand{\B}{{\mathcal B}}
\newcommand{\blue}{\color{blue}}
\newcommand{\Bor}{\mbox{\sf Bor}}
\newcommand{\Cat}{{\mathcal C}}

\newcommand{\E}{{\mathcal E}}
\newcommand{\Ev}{{\mathcal E}}

\newcommand{\Ext}{\mbox{\sf Ext}}
\renewcommand{\H}{\mbox{\bf H}}
\renewcommand{\hat}{\widehat}
\newcommand{\Hilb}{\mbox{\bf Hilb}}

\newcommand{\K}{\mbox{\bf K}}
\newcommand{\M}{{\mathcal M}}
\newcommand{\Markov}{\mbox{\bf Markov}}
\newcommand{\mc}[1]{{\mathcal #1}}

\newcommand{\N}{{\mathbb N}}

\renewcommand{\P}{{\mathbb P}}

\renewcommand{\Pr}{{\mathcal Pr}}
\newcommand{\Prob}{\mbox{\bf Prob}}
\newcommand{\R}{{\mathbb R}}
\renewcommand{\S}{S}
\newcommand{\Set}{\mbox{\bf Set}}

\newcommand{\tempout}[1]{{}}
\newcommand{\Tesp}{\mbox{\bf Tesp}}
\newcommand{\Tr}{\mbox{Tr}}
\renewcommand{\tilde}{\widetilde}
\newcommand{\ts}[1]{{\sf #1}}
\newcommand{\V}{{\mathbb V}}

\title{\sf Classical Explanations in (and of) General Probabilistic Theories}
\author{\sf John Harding\footnote{Department of Mathematics and Statistics, New Mexico State University} \footnote{JH acknowledges the support of NSF grant DMS-2231414.} and Alex Wilce\footnote{Department of Mathematics and Computer Science, Susquehanna University}}
\begin{document}

\maketitle

\begin{abstract} 
A probabilistic model consists of a test space $\M$ and a designated set of probability weights on $\M$, called the state space of the model. These structures form a category, $\Prob$, in a natural way. A probabilistic theory is then a functor from a category of (say, physical) systems and processes 
(or mathematical proxies for these) into $\Prob$.  We introduce a notion of the explanation of one model by another as a particular kind of span in $\Prob$. More precisely, an explanation of a model $A$ by a model $B$ consists of a model $C$, a quotient morphism  from $C$ to $A$, and an embedding of $B$ into $C$ (all of which terms we define).  We show that co-explanations have pullbacks, which are explanations, allowing us to compose explanations.  We then show that, subject to one very minor restriction --- that M contain no more than one single-outcome test --- every locally-finite model, and in particular, every quantum model, has a canonical classical explanation. The construction is functorial, so every locally-finite, point-closed  probabilistic theory has a canonical sharp classical extension. We explain the connections between our framework and that of ontological representations, and also consider the familiar tension between classicality and locality in our terms. Unsurprisingly, our canonical classical explanations are highly non-local. But because we do not assume finite-dimensionality, local tomography, or that all probability weights are states, some subtleties arise. In particular, there is a distinction between entanglement of states, which can occur even classically, and of probability weights, which can't. 
\end{abstract} 


\section{Introduction}

Classical probability theory contains the tacit assumption that all experiments have common refinements, and can, therefore, effectively be performed jointly. QM suggests that this is not a reasonable assumption. Attempts to generalize probability theory in order to accommodate incompatible experiments have proceeded along two main lines: the ``quantum-logic" approach, in which Boolean algebras are replaced by more general orthocomplemented posets (typically, orthomodular lattices or posets), and the ``convex" approach, in which the simplex of probability measures on a classical measurable space is replaced by a more-or-less arbitrary convex set (typically compact, or at least radially bounded) with states represented by so-called effects.  

The literature is replete with theorems showing that models of either of these types can be represented --- in one way or another --- in terms of a classical probabilistic model, provided we are willing to accept (i) failures of contextuality and (ii) loss of locality. However, these results are obtained within several different formal frameworks, and the representations they offer are not all equivalent.  Our purpose in this paper is to provide a unified, and somewhat more general, picture; one that we believe is ultimately clarifying.  

We begin by reviewing a framework for generalized probability theory that includes both the quantum logic and convex approaches as special cases. Within this, we construct a category $\Prob$ of concrete probabilistic models, following \cite{FR-Dirac, CCAM, Short-Course}. We then propose a precise definition of an {\em explanation} of one probabilistic model by another (briefly: a particular kind of span in $\Prob$). Reworking a construction first discussed in \cite{SEP}, we show that there exists an endofunctor $\Bor : \Prob \rightarrow \Prob$ sending each probabilistic model $A$ to a classical {\em Borel model} $\Bor(A)$, in terms of which $A$ has an explanation. 

We explain how this generalizes the classical embeddings and extensions of Beltrametti and Bugajski \cite{BelBug} as well as the framework of ontological models  \cite{Leifer}. We also note that any {\em probabilistic theory} --- a functor $M : \Cat \rightarrow \Prob$, where $\Cat$ is a category representing physical systems and processes \cite{Shortcut, CCAM, Short-Course} --- has a classical representation, namely, $\Bor \circ M$.  

Finally, we explore the ways in which this canonical ``classicalization" interacts with composite systems. The results are unsurprising: if $AB$ is a nonsignalling composite of $A$ and $B$, then if $AB$ admits a classical explanation in our sense, then all of its states are separable as joint probability weights, but not necessarily as joint states, of $A$ and $B$ --- a distinction we explain in Section 4.  It follows that entangled probability weights on $AB$ have no classical explanation, even in our quite generous sense. We also consider how Bell's Theorem can be formulated in this framework.  What novelty there is here consists in the level of generality (e.g., we do not assume local tomography or finite-dimensionality, and allow for a broader notion of classical explanation than is usually dealt with in the literature), and in the detailed development of the mathematics supporting this generality. In Section 5, we conclude with some quasi-philosophical reflections on the different attitudes one can reasonably take about the significance of this result for our understanding of the world.

Somewhat similar ideas are explored in an earlier paper by Gheorgiu and Heunen \cite{Gheorgiu-Heunen} in the more restricted context of ontological models of quantum theory. Specifically, Gheorgiu and Heunen consider an ontological model as a functor from $\Hilb$ to the category $\Markov$ of measurable spaces and Markov processes. 
However, \cite{Gheorgiu-Heunen} are more concerned with reproducing existing no-go theorems; our focus here is on positive results.\\


\section{Probabilistic models and theories} 

In discrete classical probability theory, a {\em probabilistic model} is usually understood to consist of the outcome-set $E$ of some experiment, together with a single probability weight thereon.  There are two obvious ways in which to generalize: one is to allow for multiple experiments (and thus, multiple outcome-sets); the other is to allow for multiple probability weights. One motivation for the former is to better understand theories like QM in which distinct statistical experiments can be incompatible. The motivation for the latter is simply to allow for some uncertainty as to the true probability weight, e.g., in situations in which many different probability weights can be prepared (either deliberately or ``by Nature".) 

These ideas motivate the following notions (see, e.g., \cite{FR-EL, FR-Dirac, Tesp}): A {\em test space} 
or {\em manual}\footnote{``Manual" is the older term; it motivates our notation, and survives in the names of some of our examples.} is a collection $\M$ of non-empty sets, understood as outcome-sets for some catalogue of possible experiments. We refer to $X = \bigcup \M$ as the {\em outcome-space} for $\M$. A {\em probability weight} on $\M$ is a function $\omega : X \rightarrow [0,1]$ summing to one on each test. We write $\Pr(\M)$ for the set of all probability weights on $\M$. Note that $\Pr(\M) \leq [0,1]^{X}$ is convex. If $\M$ is {\em locally finite}, meaning that all tests are finite, then $\Pr(\M)$ is also pointwise closed, and hence, compact. 

\begin{definition}{\em  A {\em probabilistic model} --- 
PM , for short --- is simply a pair $(\M,\Omega)$ where $\M$ is a test space and $\Omega$ is a set of probability weights on $\M$, referred to as the {\em state space} of the model.} \end{definition} 

A very simple example consists of single test $E$ and a single probability weight $\mu$ thereon. For another, consider $\M = \{\{a,b\}, \{u,v\}\}$ with $\Omega = \Pr(\M)$: it is easy to see that a probability weight on $\M$ is determined uniquely by its values on $a$ and $u$, so we have an affine isomorphism $\Pr(\M) \rightarrow [0,1]^2$. Any subset $\Omega$ of $[0,1]^2$ can serve as the state space for a model based on $\M$. We will consider further examples immediately below. At this stage, we impose no  conditions on either $\M$ or $\Omega$.

\subsection{Classical and Quantum Models} 

We now outline how standard measure-theoretic probability theory 
and its non-commutative analogue based on von Neumann algebras both 
fit into the current framework. 

\begin{example}[\bf Borel models] \label{ex: sharp Borel}
{\em If $\S$ is a topological space, the corresponding {\em Borel test space} is the collection $\B(\S)$ of finite or partitions of $\S$ by Borel sets. Thus, outcomes of $\B(\S)$ are non-empty Borel sets, and events are finite pairwise disjoint collections of Borel sets. Two events are perspective iff they have the same union; using this, it easily follows that probability weights on $\B(\S)$ are essentially the same things as finitely-additive Borel probability measures on $\S$ \footnote{Indeed, a finitely-additive Borel probability measure {\em is} simply a probability weight on $\B(\S)$, trivially extended to the entire Borel field by setting $\alpha(\emptyset) = 0$.}. Accordingly, by a {\em Borel model}, we mean any model of the form $(\B(\S),\Omega)$ where $\Omega$ is any family of Borel probability measures on $\S$.
}\end{example} 

A natural modification of this example is to consider the 
set $\B_{\sigma}(\S)$ of countable partitions of $\S$ by 
Borel sets.  Probability weights on $\B_{\sigma}(\S)$ are 
then countably-additive Borel probability measures on $\S$. 
One can also attach probabilistic models to arbitrary measurable spaces, or indeed, arbitrary Boolean algebras, in the same way. 

\begin{example}[\bf Projective quantum models] \label{ex: projective QM} 
{\em If $\H$ is a Hilbert space, let $\P(\H)$ be its lattice of projection operators. A {\em projective test} on $\H$ is simply a finite set of non-zero pairwise orthogonal projections with sum (equivalently, join) equal to $\1$. We write $\M(\H)$ for the collection of such projective tests. Thus, outcomes of $\M(\H)$ are non-zero projections, events are finite or countable pairwise disjoint families of such projections, and two events are perspective iff they have the same sum. It follows easily that probability weights on $\M(\H)$ are essentially the same things as finitely additive probability measures on $\P(\H)$. It is a corollary of Gleason's Theorem that such a measure extends uniquely to a state on the algebra $\B(\H)$ of bounded linear operators on $\H$. In particular, every density operator $W$ on $\H$ yields a probability weight on $\M(\H)$, namely $\alpha_{W}(p) = \Tr(Wp)$. By a {\em Projective Hilbert model}, we mean a PM  of the form $(\M(\H), \Omega(\H))$ where $\Omega(\H)$ is some set of probability weights arising from density operators on $\H$ in this way. 

As with Borel manuals, we can enlarge $\M(\H)$ to the set 
$\M_{\sigma}(\H)$ of countable partitions of unity in $\P(\H)$. More generally, if $\mathfrak A$ is a von Neumann algebra, then replacing $\P(\H)$ with $\P({\mathfrak A})$, we can define test spaces $\M({\mathfrak A})$ and $\M_{\sigma}({\mathfrak A})$. The generalized Gleason Theorem \cite{Bunce-Wright} again 
identifies probability weights on the former with states, 
and those on the latter with $\sigma$-additive states, 
on $\mathfrak  A$. }
\end{example} 

\label{para: unsharp models} {\bf Unsharp probabilistic models}
While our focus in this paper is on probabilistic models in which outcomes are ``sharp", meaning that we can definitively say whether or not a given outcome has occurred, there is a substantial literature on ``unsharp" measurements. A few remarks are in order here about how these fit into the present scheme.  Let $\S$ be a measurable space, and let $B(\S,\Sigma)$ denote the space of  bounded measurable real-valued functions on $\S$.  A {\em response function} on $\S$ is simply a positive 
measurable function  $f \in B(S,\Sigma)$ with $0 \leq f(s) \leq 1$ for every $s \in S$.  A finite {\em partition of unity}  in $B(S,\Sigma)$ is a list $(f_1,...,f_n)$ of response functions with $\sum_{i=1}^{n} f_i = 1$.  If 
$\mu$ is a probability measure on $(S,\Sigma)$, then 
$i \mapsto \int f_i d\mu$ defines a probability weight 
on the indexing set $\{1,...,n\}$, which one understands 
to be the probability of obtaining the $i$-th outcome in 
an ``unsharp" experiment with outcomes labelled by $\{1,...,n\}$. Notice, however, that we can have $f_i = f_j$ for distinct $i, j$. 

\begin{definition}[\bf Unsharp Borel manuals] \label{def: unsharp Borel} The {\em unsharp} Borel manual (or test space) associated 
with $(S,\Sigma)$ is the collection $\M^{o}(S,\Sigma)$ of all {\em graphs} of ordered partitions of unity in $B(S,\Sigma)$. 
\end{definition} 

Thus, outcomes in $\M^{o}(S,\Sigma)$ are pairs $(i,f)$ where 
$i \in \N$ and $f$ is a response function. If $\mu$ is a probability measure on $S$, then $\mu^{o}(i,f) = \int_{\S} f d\mu$ defines a probability weight on $\M^{o}$ We obtain in this way a probabilistic model $M^{o}(\S,\Sigma) = (\M^{o}(\S,\Sigma), \Delta^{o})$, where $\Delta^{o} = \{ \mu^{o} | \mu \in \Delta\}$.  We call this the {\em unsharp Borel model} associated with $(S,\Sigma)$.  

A similar approach allows us to construct ``unsharp" quantum models, and, more generally,  unsharp probabilistic models associated with arbitrary order unit spaces. We can  also 
apply a similar construction above to ordinary  probabilistic models, as follows. suppose $\M$ is any test space. We define the corresponding {\em ordered test space} $\M^{o}$ to consist of the graphs of all bijections $f : \{1,...,|E|\} \rightarrow E$ for every 
$E \in \M$. Outcomes of $\M^{o}$ are ordered pairs 
$(i,x)$ where $x \in X$ and $i \leq |E|$ for some test $E$ containing $x$. Any probability weight $\alpha$ on
$\M$ extends to one on $\M^{o}$ by setting 
$\alpha^{o}(i,x) = \alpha(x)$. Thus, given a PM  
$A = (\M,\Omega)$, we obtain a corresponding {\em ordered model} $A^{o} = (\M^{o}, \Omega^{o})$ where 
$\Omega^{o}$ consists of weights $\alpha^{o}$ with $\alpha \in \Omega$.  This construction will be useful later on.
For further details, see \cite{Unsharp}. 


\subsection{Categories of Probabilistic Models}

If we wish to develop a general probability theory on basis of probabilistic models, we need to identify suitable maps between them.  These should behave reasonably both with respect to the models' test-space structure, and with respect to their states.  

{\bf Test space morphisms} Given the generality and comparatively loose structure of test spaces, it should not be surprising that there is more than one plausible way to define a ``morphism" of such models \cite{FR-MMQM, Tesp}. The following is adequate for our purposes\footnote{What we are calling {\em morphisms} of test spaces, Foulis and Randall called {\em conditionings}, using the term morphism for a much broader class of mappings.}

\begin{definition} {\em Let $\M$ and $\M'$ be test spaces with outcome-sets $X$ and $X'$, respectively.  A {\em test space morphism} $\phi : \M \rightarrow \M'$ is an 
event-valued mapping $\phi : X \rightarrow \Ev(\M')$ 
such that 
\vspace{-.2in}
\begin{itemize} 
\item[(i)] $a \in \Ev(\M) \Rightarrow \phi(a) := \bigcup_{x \in a} \phi(x) \in \Ev(\M')$;  
\item[(ii)] $x \perp y$ in $X$ implies $\phi(x) \perp \phi(y)$; 
\item[(iii)] if $a, b$ are events of $\M$ with $a \sim b$, 
then $\phi(a) \sim \phi(b)$. Equivalently, 
$\phi(E) \sim \phi(F)$ for all tests $E, F \in \M$. 
\end{itemize} 
In the case in which $\phi(x) = \{y\}$ for an outcome $y \in Y$, 
we will ordinarily write $\phi(x) = y$.  }
\end{definition}

We say that a morphism is {\em test-preserving} iff $\phi(E)$ is a test in $\M'$ for every test $E \in \M$,  {\em positive} iff $\phi(x) \not = \emptyset$ for every $x \in X$, and {\em outcome-preserving} when $\phi(x)$ is a singleton for every $x \in X$. 
When $\phi$ is both positive and outcome-preserving, we naturally regard it as a mapping $X \rightarrow X'$. 
An {\em interpretation} is a  test-preserving morphism. 

For a simple example, if $\M \subseteq \M'$, then 
inclusion mapping from $X = \bigcup \M$ to $X' = \bigcup \M'$ is an interpretation. For another, consider a subset $S \subseteq X = \bigcup \M$.  We say that $S$ is a {\em support} iff the set $\M_{S} := \{E \cap S | E \in \M\}$ is irredundant, and thus a test space in its own right. In this case, the mapping $\phi : X \rightarrow \Ev(\M)$ given by $\phi(x) = \{x\}$ if $x \in S$ and 
$\phi(x) = \emptyset$ if $x \not \in S$ is non-positive interpretation morphism.  On the other hand, the inclusion mapping $x \mapsto \{x\}$ from $S$ to $X$ is generally not a test space morphism at all, as simple examples show that condition (iii) can fail.  

Turning to more specific cases, we can characterize 
morphisms between Borel and Quantum models as follows. 


\begin{example}[\bf Morphisms of Borel manuals]{\em Let $(S,\Sigma)$ and $(T,\Xi)$ be measurable 
spaces, and let $\M(S,\Sigma)$ and $\M(T,\Xi)$ be the 
corresponding manuals of countable measurable partitions of 
$S$ and $T$, respectively. 
If $f : S \rightarrow T$ is a measurable mapping, then we have $\phi := f^{-1} \Xi \rightarrow \Sigma$, a Boolean homomorphism. Let $J$ be the kernel of $\phi$, i.e., the ideal 
\[J = \{ b \in \Xi | f^{-1}(b) = \emptyset\} = 
\{ b \in \Xi | b \cap f(S) = \emptyset\}.\]
Then $\Xi \setminus J$ is a support of 
$\M(T,\Xi)$\footnote{since if $E, F$ are two partitions of $S$ by members of $\Xi$, and $E \setminus J \subsetneq F \setminus J$, then 
\[\bigcup_{a \in E \setminus J} f^{-1}(a) \subsetneq 
\bigcup_{b \in F \setminus J} f^{-1}(b).\]
The first union is all of $S$, and so must be the second, 
a contradiction. }
We have, for $a, b \in \Xi \setminus J$, that 
$f^{-1}(a), f^{-1}(b)$ are non-empty, thus elements 
of $\Sigma \setminus \{\emptyset\}$, and 
$f^{-1}(a) \cap f^{-1}(b) = \emptyset$. In other words, 
$\phi$ preserves orthogonality. If $\{a_i\}$ and $\{b_j\}$ 
are events of $A(T,\Xi)$ --- that is, pairwise disjoint 
families of non-empty sets in $\Xi$ --- then 
$\{a_i\} \sim \{b_i\}$ iff $\bigcup a_i = \bigcup b_j$ in $\Xi$, so 
$\bigcup f^{-1}(\{a_i\}) = f^{-1}(\bigcup a_i) = f^{-1}(\bigcup b_j) = \bigcup f^{-1}(\{b_j\})$, so $f^{-1}(\{a_i\}) \sim f^{-1}(\{b_j\})$. Thus, $f^{-1}$ defines a test-space morphism 
$\M(T,\Xi) \rightarrow \M(S,\Sigma)$. Note that this is 
test preserving, hence an interpretation, but is positive iff $f$ is surjective.  }
\end{example} 



\begin{example}[\bf Morphisms of projection manuals]{\em Let $\mathfrak A$ and $\mathfrak B$ be 
von Neumann algebras, with corresponding projection 
manuals $\M({\mathfrak A})$ and $\M({\mathfrak A})$ as 
in \ref{ex: projective QM}.  It is not hard to show that 
any outcome-preserving morphism $\phi : \M({\mathfrak A}) 
\rightarrow \M({\mathfrak B})$, regarded as a mapping 
$\phi : \P({\mathfrak A}) \rightarrow \P({\mathfrak B})$ (setting $\phi(0) = 0$), is {\em ortho-additive}, 
in that $\phi(a + b) = \phi(a) + \phi(b)$ whenever $a$ and $b$ are orthogonal projections; conversely, any ortho-additive mapping $\P({\mathfrak A}) \rightarrow \P({\mathfrak B})$ defines an outcome-preserving morphism 
$\M({\mathfrak A}) 
\rightarrow \M({\mathfrak B})$. In particular, 
any isomorphism or anti-isomorphism $\phi : {\mathfrak A} \rightarrow {\mathfrak B}$  defines a morphism of the 
corresponding projection manuals. As a 
special case, if ${\mathfrak A} = {\mathcal B}(\H)$ and 
${\mathfrak B} = {\mathcal B}(\K)$ for Hilbert spaces 
$\H$ and $\K$, then any isometry or anti-isometry $W : \H \rightarrow \K$ defines a morphism $\phi(a) = W a W^{\ast}$ yields such a morphism.  In \cite{Wright}, R. Wright, generalizing Dye's Theorem, showed that if 
$\mathfrak A$ and $\mathfrak B$ are von Neumann algebras 
acting on Hilbert spaces $\H_1$ and $\H_2$, and 
$\mathfrak A$ has no type $I_{2}$ summand, then any completely ortho-additive, unit-preserving map $\phi : \P(\A) \rightarrow \P(\B)$ must have the structure 
\[\phi(P) = W_{1}(P \otimes \1_{1})W_{1}^{\ast} + W_{2}(P \otimes \1_{2}) W_{2}^{\ast}\]
where $\K = \K_1 \oplus \K_2$ and 
where $W_i : \H \otimes \H_i \rightarrow \K_i$ are 
an isometry and an anti-isometry. }
\end{example}

If $\phi : \M \rightarrow \M'$ and $\psi : \M' \rightarrow \M''$ are morphisms, then 
\[(\psi \circ \phi)(x)  \ := \ \bigcup_{y \in \phi(x)} \psi(y)\]
is again a morphism. It's easily checked that this rule of composition is associative. The mapping $x \mapsto \{x\}$ provides an identity morphism from $\M$ to itself. Thus, we have a category, $\Tesp$, of test spaces and morphisms. Note that every morphism $\phi : X = \bigcup \M  \rightarrow \Ev(\M')$ lifts to a mapping $\tilde{\phi} : \Ev(\M) \rightarrow \Ev(\M')$, given by $\tilde{\phi}(a) = \bigcup_{x \in a} \phi(x)$, and then $\tilde{\psi \circ \phi} = \tilde{\psi} \circ \tilde{\phi}$; thus, we can regard $\Ev$ as a functor from $\Tesp$ to $\Set$.

{\label{para: PM morphisms} {\bf PM  Morphisms} It is straightforward that if $\phi$ is a test space morphism and $\beta$ is a probability weight on $\M'$, then $\phi^{\ast}(\beta) := \beta \circ \phi$ is a sub-normalized positive weight on $\M$. This gives us an affine mapping $\beta \mapsto \phi^{\ast}(\beta)$ from $\Pr(\M')$ to $\Pr(\M)$. Moreover, if $\M'$ is locally finite, this mapping is continuous with respect to the relative product topologies on $\Pr(\M)$ and $\Pr(\M')$.

\begin{definition} {\em If $A$ and $B$ are probabilistic models, a {\em PM -morphism} $\phi: A \rightarrow B$ is a 
test-space morphism $\phi : \M(A) \rightarrow \M(B)$ such 
that for every $\beta \in \Omega(B)$, $\phi^{\ast}(\beta)$ is a {\em possibly sub-normalized} state of $A$. }
\end{definition} 

This last condition tells us that $\phi^{\ast}$ defines an affine mapping $\Omega(B)) \subseteq \V(A)$ with $u_{A}(\phi^{\ast}(\beta)) \leq 1$ for every $\beta \in \Omega(B)$. This extends uniquely to a positive linear mapping $\phi^{\ast} : \V(B) \rightarrow \V(A)$ (it is generally harmless here to use the same symbol for this extension). This dualizes again to give a positive linear mapping $\E(\phi) : \E(A) \rightarrow \E(B)$ with $\E(\phi)(u_A) \leq u_B$. 

We shall say that a model $A$ is {\em full} iff 
$\Omega(A) = \Pr(\M(A))$. If so, then any test-space morphism $\phi : A \rightarrow B$ defines a PM -morphism.    In particular, for the full models associated with measurable spaces, or those associated with Hilbert spaces of dimension $> 2$, all test-space morphisms are PM -morphisms. However, the quantum model associated with ${\mathbb C}^2$ is not full. 

It's clear that PM -morphism compose to give PM -morphisms,  so we have a category, $\Prob$. 

\subsection{Embeddings, Quotients, and Explanations} 

We can now explain what we mean by an {\em explanation} of one model by another. We begin with two cases in which this usage seems relatively obvious.   

Let $\A$ and $\B$ be test spaces with outcome-spaces $X$ and $Y$, respectively. {An interpretation $\phi : \ \rightarrow \B$ is a {\em test-space embedding} iff $\phi$ is positive and injective on outcomes, and a {\em quotient mapping} of test spaces iff it is surjective on tests. A quotient mapping is automatically surjective on outcomes,  
and thus the adjoint mapping $\phi^{\ast} : \Pr(\B) \rightarrow \Pr(\A)$ is injective.}

\begin{definition}\label{def: PM embeddings and quotients} {\em  
A PM -morphism $\phi : A \rightarrow B$ is 
\begin{itemize} 
\item[(a)] a $PM$-{\em embedding} if it is a test space embedding of $\mc{M}(A)$ into $\mc{M}(B)$ and $\phi^{\ast}$ is surjective, that is, $\phi^{\ast}(\Omega(B)) = \Omega(A)$; 
\item[(b)] a $PM$ {\em quotient morphism} iff $q$ is a quotient mapping of test spaces with $q^{\ast}(\Omega(B)) \subseteq \Omega(A)$. 
\end{itemize} }
\end{definition} 

{\begin{example} If $(S,\Sigma_{S})$ and $(T,\Sigma_{T})$ are 
measurable spaces and $f : S \rightarrow T$ is a measurable mapping 
then $f^{-1} : \B(T) \rightarrow \B(S)$ is a quotient mapping 
if $f$ is injective and an embedding if $f$ is surjective [CHECK]
\end{example} }

If $\phi : A \rightarrow B$ is an embedding, there is an 
obvious sense in which $B$ ``explains" $A$, while if $\phi$ is a quotient mapping, then there is a sense in which $A$ explains $B$. Indeed, if $\phi : A \rightarrow B$ is an embedding, then $B$ explains $A$: we obtain $A$ from $B$ by restricting access to certain tests, and identifying states that can no longer be distinguished by this restricted set.   If $\phi$ is a quotient mapping, then $A$ explains $B$: we obtain $B$ from $A$ by neglecting to distinguish between certain outcomes of $A$, and restrict the set of states to (some of) those that do not distinguish between outcomes we've identified.

It is natural to want the relation ``explains" to be transitive. Thus, 
if $A$ is a quotient of $B$ and $B$ is embedded in $C$, we should like to say that $C$ (or rather, the quotient map and the embedding) provide an explanation of $A$. In fact, we 
will adopt this as our definition.

\begin{definition}[\bf Explanations] {\em An {\em explanation} of a model 
$A$ by a model $B$ is a triple $(C,\phi,\psi)$ where 
$C$ is model, $\phi : C \rightarrow A$ is a quotient 
mapping, and $\psi : C \rightarrow B$ is an embedding: 
\[
\begin{tikzcd} 
C \arrow[dd, "q"] \arrow[rr,"\phi"] & & B\\
& & \\
A & & 
\end{tikzcd} 
\]
}
\end{definition}

As we will see, this notion of explanation is indeed transitive. First, we address the natural question of why we do not consider the reverse (dual) construction, in which $\psi$ is the $\phi : A \rightarrow C$ is an embedding and $\psi : B \rightarrow C$ is a quotient mapping. The reason is that any ``explanation" of this form gives rise to one of the 
above type. Let us call a diagram of the form 
\[
\begin{tikzcd}
\hspace{.3in} & & & B' \arrow[dd,"q"]  \\
& & &  \\
& A' \arrow[rr, "e"] & & B & 
\end{tikzcd}
\]
(still with $q$ a quotient and $e$, an embedding) a {\em sub-quotient}. 

Thus, explanations are a particular kind of {\em span} in 
$\Prob$, while subquotients are a kind of co-span. In any 
category with pullbacks, spans compose according to the rule 
\[
\begin{tikzcd} 
Z \arrow[r,dashed] \arrow[d, dashed] 
\arrow[dr, phantom, "\scalebox{1.75}{$\lrcorner$}" , very near start, color=black]
&  Y \arrow[d] \arrow[r] & C\\
X \arrow[d] \arrow[r] &  B & \\
A &  & 
\end{tikzcd} 
\]
To make composites unique, we must choose a pullback for each pair $X, Y$. This makes composition of spans associative only up to isomorphism: spans serve as (horizontal) morphisms in a bicategory. See \cite{Johnson-Yau} for details. 

Pullbacks do not generally exist in $\Tesp$ or $\Prob$,  essentially, because they don't exist in the category 
of non-empty sets and mappings. But they do exist in 
special cases. In particular, 

\begin{proposition}\label{prop: pullbacks} Sub-quotients have pullbacks, which are explanations.
\end{proposition} 

{\em Proof:} Let $A \stackrel{e}{\rightarrow} C \stackrel{q}{\leftarrow} B$ be a sub-quotient; thus, $e : A \rightarrow C$ and $q : B \rightarrow C$ are respectively an embedding and a quotient of models. Letting $A$, $B$ and $C$ have  
$X$, $Y$ and $Z$, respectively, we can view $e$ and $f$ 
as mappings $e : X \rightarrow Z$ and $q : Y \rightarrow Z$, with $e$ injective and $q$ surjective. Define 
\[Y' = q^{-1}(e(X))\]
and let $\M'  =  \{ F \subseteq Y' | q(F) \in e(\M(A))\}$.
Thus, $\M'$ is a sub-test space of $\M(B)$, precisely the one carried by $q$ onto $\M(A)$'s image in $\M(C)$.  Note that since $q$ is a quotient, every outcome in $e(X)$ arises as 
$q(y)$ for some $y \in Y$, so $\bigcup \M' = q^{-1}(e(X)) = Y'$. Since $e$ is injective, we have an inverse mapping 
$e^{-1} : e(X) \rightarrow X$. By construction, $q' := e^{-1} \circ q|_{Y'}$ is a quotient mapping from $\M'$ onto $\M(A)$,
and $e'$ = the inclusion map $Y' \rightarrow Y$ defines 
a test-space embedding from $\M'$ to  $\M(B)$. 

Now construct a model $B' = (\M',\Omega')$ where $\Omega' = \Omega(B)|_{Y'}$, that is, $\Omega'$ consists of restrictions to $Y'$ of states of $B$. The inclusion mapping $e' : Y' \rightarrow Y$ then defines an embedding $B' \rightarrow B$. 
We need to show that $q' : Y' \rightarrow X(A)$ is also 
a PM -morphism. To this end, let $\alpha \in \Omega(A)$. 
Since $e$ is an embedding, there is some $\gamma \in \Omega(C)$ with $\alpha = e^{\ast}(\gamma)$. Let $\beta' = q^{\ast}(\gamma)|_{Y'}$, a state of $B'$. We claim that 
$\beta' = (q')^{\ast}(\alpha)$. For if $y \in Y'$ 
and $z = q(y)$, then $z \in e(X)$, i.e., there is a 
unique $x = e^{-1}(z)$ with $q(y) = e(x)$. Now we have 
\begin{eqnarray*}
(q')^{\ast}(\alpha)  
& = & \alpha \circ e^{-1} \circ q|_{Y'} \\
& = & \gamma \circ e) \circ e^{-1} \circ q|_{Y'} \\
& = & \gamma \circ q|_{Y'} \\
& = & (\gamma \circ q)|_{Y'}\\
& = & q^{\ast}(\gamma) |_{Y'}\\
& = & \beta.
\end{eqnarray*} 
This shows that $q'$ is also a morphism of models. 

To see that that $(B',e',q')$ is the pullback of $q$ and $e$, apply the $\Ev$ functor: we find that $\Ev(B') = \tilde{q}^{-1}(\tilde{e}(\Ev(A)))$. Since $\tilde{e}$ and $\tilde{q}$ are still injective and surjective, respectively, we see that $\Ev(B')$ is the pullback of these in $\Set$. Thus, given morphisms $\phi : D \rightarrow A$ and $\psi : D \rightarrow B$ with $e \circ \phi = q \circ \psi$, and regarding these as mappings $U \rightarrow \Ev(A)$  and $U \rightarrow \Ev(B)$ where $U = \bigcup \M(D)$, we have a unique mapping 
$\xi : U \rightarrow \Ev(D)$ with $\tilde{q'} \circ \tilde{\xi} = \tilde{\phi}$ and $\tilde{e'} \circ \tilde{\xi} = \tilde{\psi}$. This is given by $\psi$'s co-restriction to $\Ev(B') \subseteq \Ev(B)$ (as $\psi$'s range lies here). Since $\psi$ is a morphism of test spaces, so is $\xi$. It remains to check that $\xi$ is a PM -morphism. Let $\beta' \in \Omega'$. Then $\beta' = \beta|_{Y'}$ for 
some $\beta \in \Omega(B)$. Thus, 
\[\xi^{\ast}(\beta') = (e' \circ \xi)^{\ast}(\beta) = \psi^{\ast}(\beta);\]
since $\psi$ is a PM -morphism, $\psi^{\ast}(\beta) \in \Omega(D)$. Thus, $\xi$ is a morphism of models, and 
we are done. $\Box$

It follows that explanations compose associatively (up to a canonical isomorphism), and we can, if we wish, form a bicategory of probabilistic models and explanations. 
We will not pursue this here. For our purposes, the important points are that our notion of explanation is indeed transitive, and that quotients and embeddings are (in their dual ways) instances of explanations.

\section{Classical Explanations} 

We will now show that every locally finite probabilistic model has an explanation, and indeed a {\em canonical} one, in terms of a Borel model.  This is largely a repackaging of well-known results in the present framework, but we think this offers some slight additional generality, and a more-than-slight increment of clarity. 

\subsection{Weak-$\ast$ integrals and barycenters} 

We will need a few facts from analysis. First, recall that if $\mu$ is a finitely additive probability measure on an algebra $\Sigma$ of subsets of a set $S$, and if $f : S \rightarrow \R$ is a bounded $\Sigma$-measurable function, then $\int_{S} f d\mu$ exists, with $|\int_{S} f d\mu| \leq \|f\| \|\mu\|$ 
where $\|f\|$ and $\|\mu\|$ are the supremum and variation norms of $f$ and $\mu$, respectively.  This gives us a 
bounded linear mapping $M(S) \rightarrow B(S)^{\ast}$, 
where  $B(S)$ is the space of bounded measurable functions on $(S,\Sigma)$ and $M(S)$ is the space of bounded finitely-additive measures on $\Sigma$; in fact, one can show that 
this is a Banach-space isomorphism \cite[Theorem IV.5.1]{DS}.  Accordingly, from now on we will identify 
measures $\mu \in M(S)$ with the corresponding 
functionals, writing $\mu(f)$ for $\int_{S} f(s) d\mu(s)$ whenever convenient. 

{\bf Weak-$\ast$ integrals} $V$ is any normed  space, a function $\phi : S \rightarrow V^{\ast}$ is {\em weak-$\ast$ measurable} iff 
the function $\phi^{\ast}(v) : s \mapsto \phi(s)(v)$ is measurable for every $v \in V$. If 
$\phi$ is bounded, so is $\phi^{\ast}(v)$ for every $v \in V$, and thus 
if  $\mu$ is a bounded finitely-additive measure on $S$, we can form 
\[\hat{\mu}(\phi): v \mapsto \int_{S} \phi(s)(v) d\mu(s)\] 
which defines an element of $V^{\ast}$. We also write 
\[\hat{\mu}(\phi) = \int_{S} \phi(s) d\mu(s):\]
this is called the {\em weak-$\ast$ integral} of 
$\phi$ with respect to $\mu$.  If $\mu$ is a probability 
measure, we can regard this as the average of the functionals $\phi(s)$ with respect to $\mu$.

Now consider the case in which $V = B(S)$ and $V^{\ast} = M(S)$.  If $\delta_{s}$ is the point-mass at $s$, the mapping $s \mapsto \delta_{s}$ is bounded and weak-$\ast$ measurable (the variation norm of $\delta_s$ is $1$, and for $f \in B(S,\Sigma)$, $\delta_{s}(f) = f(s)$.)  Hence, for every probability measure $\mu$ on $S$, we can average the point masses $\delta_s$ against  $\mu$, in the weak-$\ast$ sense, to obtain 
\begin{equation}\label{eq: measures own barycenters} 
\hat{\mu} := \hat{\mu}(\delta) = \int \delta_{s} d\mu(s) \in B(S)^{\ast} = M(S).
\end{equation} 
In fact, $\hat{\mu} = \mu$. To see this, let $a \in \Sigma$ with indicator function $I_{a}$: we have 
\[\hat{\mu}(I_{a}) = \int_{S} \delta_{s}(I_{a}) d\mu(s) = \int_{S} I_{a}(s) d\mu(s) = \mu(a).\]

{\bf Barycenters} If $K$ is a compact convex subset of a topological vector space $V$. Let $A(K)$ be the space of all continuous affine (convex-linear) functionals on $K$, and let $\mu$ be a Baire 
probability measure on $K$. Then  there is a unique element $\hat{\mu} \in K$, the {\em barycenter} of $K$, with $f(\hat{\mu}) = \int_{K} f d\mu$ for all $f \in A(K)$. The mapping $\mu \mapsto \hat{\mu}$ is clearly bounded and linear. Moreover, we have the following:
\cite[Theorems 9.1 and 9.2]{Simon}:

\begin{theorem}\label{thm: cheap choquet} Every point $\alpha \in K$ can be obtained as $\hat{\mu}$ for some measure $\mu$ supported on the closure of $K$'s extreme points. 
\end{theorem} 

\tempout{
{\blue [Notation: Write $\Delta(X)$ for the (Bauer) simplex of all 
Borel probability measures on a compact Hausdorff space $X$. Note 
that the extreme points of $\Delta(X)$ are simply the point-masses 
$\delta_x$ associated with points of $x$, and that the set of 
these is weak-$\ast$ closed, hence compact, in $C(X)^{\ast}$. }
}

\subsection{Borelification} 

A version of the following construction appears in \cite{SEP}

\begin{lemma}\label{lem: DF compact}
If $\mc{A}$ is a locally finite, semi-classical space with outcome-space~$X$, then the set $\Pr(\mc{A})$ of all probability weights and the set $S$ of \ts{df} probability weights on $\mc{A}$ are both closed, hence compact, in $[0,1]^X$, and $S = \Ext(\Pr(\mc{A}))$.
\end{lemma}

\begin{proof}
Let $K = \Pr(\mc{A})$, and suppose $\alpha\in[0,1]^X$ is not in $K$. Then for some test $E \in \mc{A}$, $\sum\{\alpha(x)\mid x\in E\}\neq 1$. Hence, if $\beta\in[0,1]^X$ with $\beta(x)$ sufficiently close to $\alpha(x)$ for each $x$ in the finite set $E$, $\beta$ will also fail to be a probability weight, and thus $K$ is closed. Since $S$ is the intersection of $K$ with the closed set $\{0,1\}^X$, it is also closed. To see that $S = \Ext(K)$, note that 
since $\mc{A}$ is semiclassical, the probability weights $\alpha \in \mc{A}$ amount to indexed families $(\alpha_E)$ where $\alpha_E$ is a probability weight on $E$ for each $E\in\mc{A}$. If one of these, say $\alpha_F$, is a non-trivial convex combination $\alpha_F=\lambda \beta_F+(1-\lambda)\gamma_F$ where $\beta_F$ and $\gamma_F$ are 
probability weights on $F$, then define $\beta = (\beta_{E})$ by setting $\beta_{E} = \alpha_{E}$ for 
all $E \not = F$, with $\beta_{F}$ the given probability 
weight. Define $\gamma = (\gamma_{E})$ similarly. Then  $\alpha =\lambda \beta + (1-\lambda)\gamma$ is also a non-trivial convex combination. So if $(\alpha_E)$ belongs to $Ext(K)$, each  $\alpha_E$ must be an extreme point of the set of probability weights on $E$, and hence, a point mass, and thus $\alpha$ is \ts{df}. It is clear that  \ts{df} probability weight is an extreme point. 
\end{proof}


We are now going to show that {\em almost} every semi-classical test space $\A$ can be embedded into a Borel test space. The only potential obstacle is that if 
$x$ and $y$ are distinct outcomes with $\{x\}, \{y\} \in \A$, then every \ts{df} state assigns probability one to 
both. 

\begin{theorem}
\label{thm: sc submodel of classical}
Every locally finite, semiclassical PM $A=(\mc{A},\Omega)$ having at most one single-outcome test, has a canonical embedding into a classical model. 
\end{theorem}

\begin{proof}
Let $S$ be the set of d.f. probability weights on $\mc{A}$.  By Lemma \ref{lem: DF compact}, $S$ and $\Pr(\mc{A})$ are compact with $S = \Ext(\Pr(\mc{A})$. Let $\Sigma$ be the Borel algebra of $S$ and $\mc{B}(S)$ be the  corresponding Borel test space. That is, 
$\mc{B}(S)$ consists of finite partitions of $S$ by Borel sets.  Define a map $\phi$ from the outcomes of $\mc{A}$ to the outcomes of $\mc{B}$ --- the non-empty Borel sets --- by setting 
\[
\phi(x)\,\,=\,\,\{s\in S\mid s(x)=1\}.
\]
Note that $\phi(x)$ is the intersection of a closed set of $[0,1]^X$ with $S$, hence, is indeed a Borel set of $S$. 
Now, if $x \perp y$, $\phi(x) \cap \phi(y) = \emptyset$, 
and $\{\phi(x) | x \in E, \phi(x) \not = \emptyset\}$ 
is a partition of $S$. It is straightforward 
that if $a, b \in \Ev(\mc{A})$ with $a \sim b$, $\bigcup \phi(a) = \bigcup \phi(b)$, 
so $\phi(a) \sim \phi(b)$ in $\mc{B}(S)$. Thus, 
$\phi$ is a test-preserving test-space morphism. This much holds regardless of $\mc{A}$-s structure. But since $\mc{A}$ is semiclassical, we also have $\phi(x) \not = \emptyset$ 
for every $x \in X$. If $x$ and $y$ are distinct outcomes, 
at most constitutes a one-outcome test, so there exists 
some $s \in S$ with $s(x) \not = s(y)$. Thus, $\phi$ is 
injective on outcomes, and thus, a test-space embedding. 
Finally, suppose $\alpha\in\Omega$. By Theorem \ref{thm: cheap choquet}, $\alpha$ is the barycenter, $\hat{\mu}$, of a regular Borel probability measure $\mu$ on $S$. This means that for the continuous linear functional $\hat{x}:\R^X\to\R$ of evaluation at $x$, we have 
\[
\hat{x}(\alpha)\,\,=\,\,\int\nolimits_S\,\hat{x}\,\,d\mu.
\]
For $s\in S$ we have  $\hat{x}(s)=1$ iff $s\in\phi(x)$ and is 0 otherwise, so $\int_S\hat{x}\,d\mu=\mu(\phi(x))$. Thus 
\[
\alpha(x)\,\,=\,\,\mu(\varphi(x)),
\]
or, equivalently, $\alpha = \varphi^{\ast}(\mu)$. 
 Taking $\Omega'$ to be the set of all probability measures $\mu$ on  $(S,\Sigma)$ with $\hat{\mu}$ in $\Omega$, and letting $B = (\mc{B},\Omega')$,  $\phi$ is a PM-embedding in the sense of Definition \ref{def: PM embeddings and quotients}.  
\end{proof}

\noindent {\bf Semiclassical covers} We have just seen that locally finite semiclassical PMs can always be interpreted  as sub-models of classical models. On the other hand, every locally finite PM is a ``quotient'' of a locally finite semi-classical model in an obvious way (\cite[p. 184]{RF-Operational}, see also \cite[Section 6.2]{SEP}). Let $\mc{A}$ be a test space. For 
each $E \in \mc{A}$, let $\tilde{E}\ =\ \{(x,E)\mid x\in E\}$. The \emph{semiclassical cover} of $\mc{A}$ is the test space $\tilde{\mc{A}}$ given by
\[
\tilde{\mc{A}}\,\,=\,\,\{\tilde{E}\mid E\in\mc{A}\}.
\]
We write $\tilde{X}$ for $\bigcup \mc{A} = \{ (x,E) \mid x \in E \in \mc{A}\}$. 
Every probability weight $\alpha \in \Pr(\mc{A})$ gives rise to a probability weight $\tilde{\alpha}$ on $\tilde{\mc{A}}$ defined by $\tilde{\alpha}(x,E) = \alpha(x)$. By the  semiclassical cover of a PM $A = (\mc{A},\Omega)$, we 
mean the model $\tilde{A} = (\tilde{\mc{A}},\tilde{\Omega})$ where 
\[\tilde{\Omega} = \{ \tilde{\alpha} \mid \alpha \in \Omega\}.\] It is then straightforward that the mapping $q : \tilde{X} \rightarrow X$ given by $q(x,E) = x$ is a quotient morphism of models. 

Thus, every PM is a quotient of its semi-classical cover. Since every locally finite semiclassical model has a canonical classical embedding, we have the 

\begin{theorem}\label{thm: Borelification} Every locally finite PM $A$ with at most one one-outcome test, has a canonical classical explanation 
\[
\begin{tikzcd} 
\tilde{A} \arrow[rr, "\phi"] \arrow[dd, "q"] & & \B(S)\\
& & \\
A & & 
\end{tikzcd}
\]
where $S$ is the set of dispersion-free probability weights 
on $\tilde{A}$, understood as a subspace of $[0,1]^{\tilde{X}}$.
\end{theorem} 

{\bf Borelification as a functor} Let us write 
$\Bor(A)$ for the Borel model arising in the canonical classical 
explanation of a probabilistic model $A$. Thus, 
$\Bor(A)$'s test space is the space of finite partitions 
of the set $S$ of dispersion-free states on $\tilde{S}$ by 
Borel sets, and $\Omega(\Bor(A))$ is the set of finitely-additive Borel probability measures on $S$ having barycenter in 
$\Omega(A)$. We will refer to $\Bor(A)$ as the {\em Borelification} of $A$. It is not difficult to show that $A \mapsto \Bor(A)$ is the object part of an endofunctor $\Bor : \Prob \rightarrow \Prob$. 

As discussed elsewhere \cite{LTS, Short-Course}, one can view a {\em probabilistic theory} as a functor $F : \Cat \rightarrow \Prob$, where $\Cat$ is a category representing systems and processes. The idea is that to each system $A$, $M$ attaches a concrete probabilistic mode, $M(A) = (\M(A),\Omega(A))$, and to each process $f : A \rightarrow B$, $M$ attaches a morphism $Mf : M(A) \rightarrow M(B)$ of probabilistic models.  Thus, if $M : \Cat \rightarrow \Prob$ is any probabilistic theory, we have a canonical ``Borelification" $\Bor \circ M$ from $\Cat$ to the category of Borel models and morphisms between these.

\subsection{Ontological Representations} 
A related formulation of these ideas can be given in terms 
of so-called {\em ontological representations} \cite{Gheorgiu-Heunen, Leifer, Spekkens}. Adapted to our current language, the definition 
is as follows:

\begin{definition}[\bf Ontological representations] \label{def: ontological representation} {\em An {\em ontological representation} of a PM $A$ is measurable space $(S,\Sigma)$ of ``ontic states''  with a function $p$ assigning to each $s \in S$ and each test $E \in \mc{A}$ a probability weight $x \mapsto p(x\,|\,E,s)$ on $E$ so that  
\begin{itemize}
\vspace{-.1in}
 \item[(a)] For each $x \in E$, the function 
$s \mapsto p(x|E,s)$ is measurable, and 
\item[(b)] For each $\alpha \in \Omega_A$, there is a  probability measure $\mu$ on $S$ with 
\[\alpha(x) \,\,=\,\, \int_S p(x\, |\, E, s)\, d\mu.\] 
\end{itemize}
\vspace{-.2in}
If $p$ takes only values $0,1$ we say the representation is {\em sharp} and if $p(x\,|\,E,s)$ is independent of $E$ for every $x$ and $s$, it is {\em non-contextual}. 
}
\end{definition}

Perhaps the simplest example is the {\em Beltrametti-Bugajski representation} \cite{BelBug, Holevo, Leifer}: 
given a model $A = (\M,\Omega)$ with $\Omega$ compact, let 
$S$ be the set of extreme points of $\Omega$, i.e., 
the pure states of $A$, with Borel structure inherited 
from $[0,1]^{X(A)}$. For every $x \in E \in \M$, 
and every $\alpha \in S$, let $p(x|E,\alpha) = \alpha(x)$. 
This is evidently non-contextual. As long as $\Omega$ 
is separating, it yields an embedding of 
$A$ into the unsharp model associated with $S$. 

More generally, an ontological representation of $A$ 
based on a measurable space $(S,\Sigma)$ yields 
an embedding of $\tilde{A}$ into the unsharp Borel model 
$\B^{o}(S,\Sigma)$ associated with $(S,\Sigma)$, as discussed in 
Section \ref{para: unsharp models}.

Recall that the tests of this model consist of (graphs of) lists of positive bounded measurable functions summing to one, and that the states are induced by finitely additive probability measures on $(S,\Sigma)$ in an obvious way.
Given an ontological representation $p$ of $A = (\A,\Omega)$, define $\phi : \tilde{X} \rightarrow B(S)$ by $\phi(x,E)(s) = p(x|E,s)$. This extends uniquely to a morphism $\phi^{o} : \tilde{A}^{o} \rightarrow \B^{o}(S)$, which will be injective iff $\Omega$ is a separating set of states. 
For any probability measure $\mu$ on $(S,\Sigma)$, 
\[{\phi^{o ~\ast}}(\mu)(x,E) = \hat{\mu}(f_{x,E}) = 
\int_{S} f_{x,E}(s) d\mu(s),\]
so condition (c) tells us that $\phi^{\ast}$ is 
surjective on states.}  Thus, we have an embedding $\phi^{o} : (\tilde{A})^{o} \rightarrow \M^{o}(B(S,\Sigma))$, where $B(S,\Sigma)$ is the order-unit space of bounded measurable real-valued functions on $(S,\Sigma)$, as in \ref{para: unsharp models}. 
Moreover, we have a natural 
quotient mapping $q : (\tilde{A})^{o} \rightarrow A$, given by 
$q(i,(x,E)) = x$. Thus, we have an explanation 
\[\begin{tikzcd}
\left ( \tilde{A} \right )^{o} \arrow[dd, "q"] \arrow[rr, "\phi^{o}"] & & M^{o}(B) \\
& & \\
A & & \\
\end{tikzcd}\]
The unsharp classical model $M^{o}(B(S,\Sigma))$ can in turn be explained in terms of the sharp classical model 
$\Bor({M^{o}(B(S,\Sigma)))}$. 
 
\tempout
{Suppose now that $p$ is an ontological representation of a PM $A$ with $(S,\Sigma)$ its set of ontic states. The outcomes of the semiclassical cover of $A$ are pairs $(x,E)$ with $x \in E \in \mc{A}$. By part (b) of Definition \ref{def: ontological representation}, for each such pair $(x,E)$ we have a measurable 
function $f_{x,E}:S\to\R$ given by $f_{x,E}(s)=p(x\,|\,E,s)$.} 
In fact, the mapping $(x,E) \mapsto f_{x,E}$ is an embedding. 
By (c), if $f_{x,E} = f_{y,F}$, then 
$\tilde{\alpha}(x,E) = \tilde{(y,F)}$ for all
$\tilde{\alpha} \in \Omega(\tilde{A})$, i.e., 
$\alpha(x) = \alpha(y)$ for all $x, y \in X_A$. We are (or should be) assuming a separating set of states here, so $x = y$. Condition (a) tells us that this is test-preserving, so we have an embedding of test spaces. Finally, for any probability measure $\mu$ on $(S,\Sigma)$, 
\[\phi^{\ast}(\mu)(x,E) = \hat{\mu}(f_{x,E}) = 
\int_{S} f_{x,E}(s) d\mu(s),\]
so condition (c) tells us that ${\phi^{o}}^{\ast}$ is surjective on states. 

Since $A$ is a quotient of its semiclassical cover, this yields a unsharp classical explanation of $A$. Conversely, any unsharp classical explanation of $A$ obtained by embedding its semiclassical cover $\tilde{A}$ into an unsharp classical model, can be seen to yield an ontological representation. 

In summary: a sharp, respectively unsharp, contextual ontological representation of a probabilistic model $A$ is simply an embedding $\phi : \tilde{A} \rightarrow C$ where $C$ is an sharp, respectively unsharp, Borel model. Thus, we can paraphrase Theorem \ref{thm: Borelification} as saying that every probabilistic model arises as the quotient of one admitting a sharp ontological model. 

\section{Composition and Locality} 

By Theorem~ \ref{thm: Borelification},  every locally finite PM admits a classical explanation, hence, by the discussion in Section 3.3, 
 it admits a sharp ontological representation. The situation for composite systems is more delicate: while such a composite, if locally finite, admits an ontological representation, this will generally not satisfy a reasonable locality criterion with respect to the component models.  We return to this below.

\subsection{Composites and entanglement}  We begin with an overview of the theory of non-signalling composites of probabilistic models.  For a more leisurely account, see \cite{Tesp}. 

Let $A$ and $B$ be test spaces, with outcome-spaces $X(A)$ and $X(B)$. We abuse notation slightly to write $\M(A) \times \M(B)$  for the set of all {\em product tests} $E \times F$, with $E \in \M(A)$ and $F \in \M(B)$. We understand these to represent experiments in which $E$ and $F$ are performed on separate systems, and their outcomes, collated. Note that the outcome-space  of $\M(A) \times \M(B)$ is $X(A) \times X(B)$. A {\em joint probability weight} $\omega$ on $\M(A)$ and $\M(B)$ is simply a probability weight on $\M(A) \times \M(B)$. We 
say that $\omega$ is {\em non-signalling} iff the {\em marginal weights} 
\[\omega_{1}(x) = \sum_{y \in F} \omega(x,y) \ \ \mbox{and} \ \ 
\omega_2(y) = \sum_{x \in E} \omega(x,y)\]
are independent of the choice of $E \in \M(A)$ and $F \in \M(B)$. In this case, we define the {\em conditional states} $\omega_{1|y} \in \Pr(\M(A))$ and $\omega_{2|x} \in \Pr(\M(B))$ by 
\[\omega_{1|y}(x) = \omega(x,y)/\omega_{2}(y) \ \ \mbox{and} \ \ 
\omega_{2|x}(y) = \omega(x,y)/\omega_{1}(x),\]
respectively. If $\omega_{1|y} \in \Omega(A)$ and $\omega_{2|x} \in \Omega(A)$ for every $x \in X(A)$ and $y \in X(B)$, then we say that $\omega$ is a non-signalling joint {\em state} of $A$ and $B$.  We write $\Omega_{NS}(A,B)$ for the set of all non-signalling joint states of $A$ and $B$. The {\em minimal non-signalling composite}\footnote{For {\em cognoscenti}: the state space of $A \times B$ is the {\em maximal} tensor product of the state spaces of $A$ and $B$; thus, the 
effect cone of $A \times_{NS} B$ is the minimal 
tensor product of those of $A$ and $B$} of $A$ and $B$ is then 
\[A \times_{NS} B \ := \ (\M(A) \times \M(B), \Omega_{NS}(A,B)).\]

The simplest example of a non-signalling probability weight 
on $A$ and $B$ is a {\em product weight}, $\alpha \otimes \beta$, given by 
\[(\alpha \otimes \beta)(x,y) = \alpha(x)\beta(y)\]
for all $x \in X(A)$ and $y \in X(B)$, where $\alpha$ and $\beta$ are probability weights on $A$ and $B$, respectively. If $\alpha$ and $\beta$ are states, then of course we speak of $\alpha \otimes \beta$ as a product state on $A$ and $B$. 


\begin{definition} {\em A {\em separable joint state} of $A$ and $B$ is a joint state lying in the closed convex hull of the product states on $A$ and $B$. A joint state that is not separable is an {\em entangled joint state}.}
\end{definition} 


Note that the distinction between separability and entanglement 
of {\em states} depends on the models, $A$ and $B$, and, in particular, on their state spaces as well as their test spaces.  If we replace $\Omega(A)$ and $\Omega(B)$ by the full sets $\Pr(A)$ and $\Pr(B)$ of all probability weights, we will speak of separable and entangled {\em weights} on $A$ and $B$.  Thus, a joint state can be separable as a joint probability weight --- that is, separable with respect to $\Pr(A)$ and $\Pr(\B)$ --- while being entangled as a joint state.  This is even possible for classical models: 

\begin{example}[\bf Classical entanglement]\label{ex: classical entanglement} {\em Let $\B$ be the test space of finite partitions of $[0,1]$ by Borel sets. Outcomes of $\B$ are non-empty Borel sets, and probability weights on $\B$ are simply finitely-additive Borel probability measures on $[0,1]$, as restricted to non-empty Borel sets. We will ignore this tiny distinction, and simply identify $\Pr(\B)$ with the simplex $\Delta$ of all finitely additive Borel measures on $[0,1]$.  Now let $\lambda$ denote Lebesgue measure on $[0,1]$, and let $\Delta_{\lambda}$ be the collection of all Borel measures on $[0,1]$ absolutely continuous with respect to $\lambda$. This is a closed face of $\Delta$, and thus, itself also a simplex. 
Now consider the probabilistic model $B = (\B,\Delta_{\lambda})$.  Every separable probability weight on $B \times_{NS} B$ --- that is, every convex combination of product states, and every pointwise limit of such states in $[0,1]^{\Sigma \times \Sigma}$ --- is absolutely continuous with respect to Lebesgue measure $\lambda \otimes \lambda$ on $[0,1]^{2}$. Let $\omega$ be the measure on $[0,1]^2$ given by $\omega(a,b) = \lambda(a \cap b)$. This is not 
absolutely continuous with respect to $\lambda \otimes \lambda$, but is non-signalling, with conditional weights that do belong to $\Delta_{\lambda}$.  Hence, $\omega$ is a non-separable, non-signalling state. On the other hand, every finitely-additive measure on $[0,1]^2$ has the form $\mu = \int_{s} \delta_{s} \otimes \delta_{t} d\mu(s,t)$ 
and is therefore separable with respect to the set $\Pr(\M_{A})$ 
of all probability weights. 
}\end{example} 

{\bf General non-signalling composites} The minimal non-signalling composite $A \times_{NS} B$ allows only product tests.  To allow for a greater abundance of joint tests, it's natural to introduce the following 

\tempout{
{\blue \begin{definition}[Alt] A {\em composite} of test spaces 
$\A$ and $\B$ is a test space $\Cat$, together with a positive, 
injective interpretation $\pi : \A \times \B \rightarrow \Cat$.  We 
say that a probability weight $\gamma$ on $\Cat$ is {\em non-signaling} iff $\pi^{\ast}(\gamma)$ is non-signaling, {\em entangled} iff 
$\pi^{\ast}(\gamma)$ is entangled, and {\em separable} iff 
$\pi^{\ast}(\gamma)$ is separable.  If $\A = \M(A)$ and $\B = \M(B)$ 
for models $A$ and $B$, we say that a model $C = (\C,\Gamma)$ is 
a {\em non-signaling composite} of $A$ and $B$ iff 
every state in $\Gamma$ belongs to $\Omega_{NS}(A,B)$. 
\end{definition} }
}

\begin{definition} {\em A {\em non-signalling composite} of $A$ and $B$ is a probabilistic model $AB$, together with a positive, injective, interpretation of models $\pi : A \times_{NS} B \rightarrow AB$. }
\end{definition} 

Thus, every pair of outcomes $x \in X(A)$ and $y \in X(B)$ correspond to an event $\pi(x,y) \in \Ev(AB)$, and every state in $\Omega(AB)$ pulls back to a non-signalling joint state $\pi^{\ast}(\omega) \in \Omega_{NS}(A,B)$, given by 
\[\pi^{\ast}(\omega)(x,y) = \omega(\pi(x,y)).\]
In a slight abuse of terminology, we will refer to $\pi^{\ast}(\omega)(x,y)$ as the {\em restriction} of $\omega$ to $A \times_{NS} B$. 

When $\pi^{\ast}$ is injective, so that every state $\omega \in \Omega(AB)$ is determined by its restriction  
joint state $\pi^{\ast}(\omega)$ to $A \times_{NS} B$, we say that $(AB,\pi)$ is {\em locally tomographic}. When every product state $\alpha \otimes \beta$ arises as the restriction $\pi^{\ast}(\omega)$ of some state $\omega \in \Omega(AB)$, we say that $(AB,\pi)$ is {\em strong}. 

In more physical terms, a composite is locally tomographic iff local measurements suffice to determine the global state, and strong, if all pairs of local states can be independently prepared (possibly in multiple ways) as the marginals of a suitable state of the global composite system. Where either of these properties is lacking, the definition of entanglement requires some care. 

\begin{definition}\label{def: separable joint states} Let $(AB,\pi)$ be a non-signaling composite 
of $A$ and $B$. A joint state of the form $\pi^{\ast}(\omega)$ 
for some $\omega \in \Omega(AB)$ is {\em preparable} in $AB$. 
A state $\omega \in \Omega(AB)$ is {\em separable} iff $\pi^{\ast}(\omega)$ lies in the closed convex hull of the preparable product states. 
\end{definition} 

If $(AB,\pi)$ is locally tomographic, then $AB$ is strong iff 
$\Omega(AB)$ contains all product states.  If $(AB,\pi)$ is strong, then $\omega \in \Omega(AB)$ is separable iff $\pi^{\ast}(\omega)$ is a separable joint state in the sense of Definition \ref{def: separable joint states}. But non-locally tomographic and non-strong 
composites do arise, even in quantum theory, as we now illustrate. 

\begin{example}[\bf Quantum Composites] \label{ex: quantum composites}
{\em Let $A$ and $B$ be the full projective quantum model corresponding to Hilbert spaces $\H_{A}$ and $\H_B$ --- both real, or both complex --- as in Example \ref{ex: projective QM}. Thus, $\M(A) = \M(\H_{A})$ consists of all maximal pairwise orthogonal projections on $\H_A$, and $\Omega(A) = \Omega(\H_A)$ is the set of probability weights $\alpha_{W}$ corresponding to density operators $W$ on $\H_A$, according to $\alpha_{W}(p) = \Tr(Wp)$, and similarly for $B$. The mapping $\pi : \M(\H_A) \times \M(\H_B) \rightarrow \M(\H_A \otimes \H_B)$ given by $\pi(p,q) = p \otimes q$ is easily seen to be a test-space embedding. 
To obtain a non-signalling composite $AB$, we must 
equip $\M(AB) = \M(\H_A \otimes \H_B)$ with a suitable state-space. 
There are multiple options. For instance, 
\vspace{-.1in}
\begin{itemize}
\item[(a)] One might take $\Omega(AB)$ to be the 
full set $\Omega(\H_{A} \otimes \H_{B})$ of density operators 
on $\H_{A} \otimes \H_{B}$,  as in standard QM.  
\item[(b)] One could take $\Omega(AB)$ to be the set of 
separable bipartite states. 
\item[(c)] If $\H_{A} = \H_{B}$, 
one could take $\Omega(AB)$ to consist of Bosonic or Fermionic mixed states --- that is, density operators $W$ on 
$\H_{A} \otimes \H_{B}$ satisfying $\sigma W \sigma = W$, or those 
satisfying $\sigma W \sigma = -W$, where $\sigma$ is the swap operator $\sigma : x \otimes y \rightarrow x \otimes y$.
\end{itemize} 
As is well-known, the composite in (a) is locally tomographic 
if $\H_{A}$ and $\H_{B}$ are complex, but not if they are real. 
In both (a) and (b), the composites are strong, but in (b), there 
are, by fiat, no entangled states. Neither of the composites in (c) 
are strong. In the Fermionic case, no product state is preparable, 
so all states are entangled; in the Bosonic case, products of 
identical states are preparable, and mixtures of these are separable. 
}\end{example}

\subsection{Bell-locality} 

We now come to the crucial idea of Bell-locality. The following is a direct paraphrase of the standard definition in terms of ontological models (e.g., \cite{Leifer}) into our language:

\begin{definition}\label{def: Bell-local} {\em  Let $\pi : A \times_{NS} B \rightarrow AB$ be a non-signalling composite and $\phi : \tilde{AB} \rightarrow B(S,\Sigma)$ be a test-space embedding into the finitely-additive Borel model associated with a measurable space $(S,\Sigma)$. We say that $\phi$ is {\em Bell-local} with respect to $\pi$ iff for each point mass $\delta_s$, $s \in S$, the probability weight $\gamma_{s}$ with $\tilde{\gamma_{s}} := \phi^{\ast}(\delta_{s})$ on $AB$ restricts to a product weight, that is, there are probability weights $\alpha_s$ on $\M(A)$ and $\beta_{s}$ on $\M(B)$ with $\pi^{\ast}(\gamma_{s}) = \alpha_{s} \otimes \beta_{s}$. }
\end{definition}

It is important to notice here that $\gamma_s$, $\alpha_s$, and $\beta_{s}$ need not be {\em states} of $AB$, $A$, or $B$. Thus, locality, as defined above, is entirely a property of the test-spaces $\M(A)$, $\M(B)$, and $\M(AB)$, and the test-space embeddings $\pi$ and $\phi$: the chosen state spaces $\Omega(A), \Omega(B)$ and $\Omega(AB)$ do not enter into it. Or, equivalently, the definition refers only to 
full models, in which $\Omega(A) = \Pr(\M(A))$, $\Omega(B) = \Pr(\M(B))$, 
and $\Omega(AB) = \Pr(\M(AB))$. 

At the same time, even if all probability {\em weights on} $\M(AB)$ 
are separable, it is possible for $\Omega_{AB}$ to contain states that are {\em not} convex combinations of product {\em states} in $\Omega_A$ and $\Omega_B$.  These would then be entangled as states, though not as probability weights, as in Example \ref{ex: classical entanglement}.

Returning to Definition \ref{def: Bell-local}, suppose that the models $A$, $B$ and $AB$ therein are locally finite, so that their state spaces are compact. For every probability weight $\mu$ on $S$, we can write 
\[\mu = \int_{S} \delta_{s} ~d\mu(s)\]
as in Equation \ref{eq: measures own barycenters}. Since $C$ is locally finite, $\pi^{\ast} : \Pr(C) \rightarrow \Pr(AB)$ is continuous (see remarks above in \ref{para: PM morphisms}),  and thus, Borel measurable. Since by Bell locality 
$\pi^{\ast} \circ \phi^{\ast}(\delta_{s}) = \alpha_s \otimes \beta_s$, we have 
\[\pi^{\ast} \phi^{\ast}(\mu)  = \int_{S} \pi^{\ast} \phi^{\ast}(\delta_{s})\ d\mu(s) = \int_{S} \alpha_{s} \otimes \beta_{s} \ d\mu(S)\]
which
belongs to the closed convex hull of the product weights, i.e., to the set of separable probability weights, on $A \times B$.  Since $\phi \circ \pi$ is an embedding, it follows that all states of $AB$ are separable {\em as probability weights}; that is, they are all (integral) convex combinations of weights restricting to products of probability weights. In summary, we have the following

\begin{proposition}\label{prop: local implies separable} If a non-signalling composite $\pi : A \times_{NS} B \rightarrow AB$ of locally finite probabilistic models admits a Bell-local classical embedding, then the restrictions of all states of $AB$ to 
$A \times_{NS} B$ are separable as joint probability weights. 
\end{proposition} 

For a given classical embedding $\phi$ to be local, therefore, all states on $AB$ must be separable as probability weights. If $AB$ supports even one state that is not separable in this sense, then $AB$ admits no local classical embedding. 

Since $\tilde{\pi} : \tilde{A} \times_{NS} \tilde{B} \rightarrow \tilde{AB}$ is a non-signalling composite, Proposition \ref{prop: local implies separable} applies to the canonical classical embedding of $\tilde{AB}$ into the Borel test space over the set $S$ of dispersion-free probability weights on $\tilde{AB}$. 
It is not difficult to see that a dispersion-free probability weight on $\M(A) \times \M(B)$ is non-signalling iff it is a product weight. Thus unless $\M(A)$ and $\M(B)$ consist of single tests, the vast majority of the d.f. probability weights $\phi^{\ast}(\delta_{s})$ arising from points $s \in S$ are signalling.  The classical embedding, then,  represents non-signalling states of $A \times_{NS} B$ as weighted averages of dispersion-free, {\em signalling}, probability weights. As a simple but instructive illustration of this, consider $A = B = \{\{x,y\}, \{u,v\}\}$. The d.f. joint probability weight defined by 
\[\omega(x,x) = \omega(x,u) = \omega(u,x) = \omega(u,v) = 1\] 
is signalling (since, e.g., $\omega_{2|E}(u) = 1$ while $\omega_{2|F}(u) = 0$). Similarly, 
\[\omega'(y,y) = \omega(y,v) = \omega(v,y) = \omega(v,u) = 1\] 
is signalling. But $\frac{1}{2}(\omega + \omega')$ is the well-known entangled PR-box state \cite{PR}. ,

\begin{remark}{\em  It is not difficult to extend Definition \ref{def: Bell-local} to apply to arbitrary classical explanations $(C,q,e)$ of a composite $AB$. One can also amend it to require that $\pi^{\ast}(\delta_{s}) = \gamma_{s}$ have the form $q^{\ast}(\alpha_{s} \otimes \beta_{s})$ for {\em states} $\alpha_s$ and $\beta_s$ in $\Omega(A)$ and $\Omega(B)$, respectively (in which case, of course, $\gamma_{s} \in \Omega(C)$). In that case, we might say that the explanation is {\em state-local}. The same argument as give above then shows that if $AB$ has a state-local classical explanation, its states are all separable {\em as states}. }
\end{remark}

\section{Conclusion}  Bell's Theorem and subsequent tests of the 
BCHSH inequality establish that one can prepare joint states 
of quantum systems that do not have local classical explanations, 
even in our broad sense. In some quarters, the significance of 
this is a matter of contention. According to one 
point of view, Bell's Theorem sets up a tension between locality and (some version of) ``classicality".  According 
to another, the theorem simply establishes that the world 
is non-local, full stop. See \cite{Maudlin, Werner} for 
a particularly sharp exchange of views on this matter. 
However, there nothing sacred about our definition of classicality, or that of locality. What has been shown above is that {\em one} notion of locality is in conflict with {\em one} notion of classicality. The following attitudes all seem to us consistent with this, and 
reasonable on their own terms (though we lean towards the first one).  

(1) The non-product dispersion-free probability weights on $\tilde{AB}$ are signalling, and hence, unphysical.  The world is non-classical, since non-commuting quantum observables are {\em in fact} not jointly testable. Whether it is local or not depends on whether or not one understands entanglement {\em per se} to be a form of non-locality. If 
so, the world is non-local {\em because} it is non-classical. 


(2) The dispersion-free probability weights on $\tilde{AB}$ provide a coherent ontology of the world; the fact that the states we observe are all non-signalling reflects a degree of fine-tuning that is simply part of the world's architecture. The world is classical, but non-local. 

(3)  Locality is part of what we {\em mean} by classicality. The world is non-classical {\em because} it is non-local. 



\begin{thebibliography}{99}

\bibitem{LTS} H. Barnum, M. Graydon, and A. Wilce, 
Composites and categories of euclidean Jordan algebras, 
Quantum {\bf 4} (2020) 

\bibitem{LTS} H. Barnum, M. Graydon, and A. Wilce, 
Locally tomographic shadows, in S. Mansfield, B. Valiron, V. Zamdzhiev (eds.): Quantum Physics and Logic 2023 (QPL 2023)
EPTCS {\bf 384} (2023),

\bibitem{BelBug} E. Beltrametti and S. Bugajski, A classical extension of quantum mechanics, J. Phys. A {\bf 28} (1995)

\bibitem{Bunce-Wright} J. L. Bunce and J. D. M. Wright, 
The Mackey-Gleason problem, Bull. Amer. Math. Soc. {\bf 26} (1992) 

\bibitem{DS} N. Dunford and J. Schwarz, Linear Operators, vol. I., Wiley, 1988


\bibitem{FR-EL} D. Foulis and C. Randall, The empirical logic approach to the physicl sciences, in A. H\"{a}rtkamper and 
H. Neumann (eds.), {\em Foundations of Quantum Mechanics in Ordered Linear Spaces}, Lecture Notes in Physics {\bf 29}, 
Springer, 1974



\bibitem{FR-MMQM} D. Foulis and C. Randall, Manuals, morphisms and quantum mechanics, in A. R. Marlowe (ed.), {\em Mathematical Foundations of Quantum Theory}, AP 1978

\bibitem{RF-Operational} D. Foulis and C. Randall, 
The operational approach to quantum mechanics, in C. A Hooker (Ed.), {\em Physical Theory as Logico-0perational Structure}, 
University of Western Ontario Series in Philosophy of Science, Reidel, 1979
Reidel, 

\bibitem{FR-Dirac} D. Foulis and C. Randall, Dirac Revisited, in 
P. Lahti and P. Mittelstaedt, eds., {\em Symposium on the foundations of modern physics. 50 years of the Einstein-Podolsky-Rosen Gedankenexperiment}, World Scientific, 1985.

\bibitem{FGR} D. Foulis, R. Greechie, and G. R\"{u}ttimann, Logico-algebraic structures II: supports in test spaces, Int. J. Theor. Physics {\bf 32} (1993)
 
\bibitem{Gheorgiu-Heunen} A. Gheorgiu and C. Heunen,  Ontological models for quantum theory as functors, in B. Coecke and M. Leifer (eds.), {\em Quantum Physics and Logic 2019 (QPL)} EPTCS {\bf 318}, 2020, pp. 196–21  {\tt arXiv: 1905.09055} 

\bibitem{Holevo} A. Holevo, {\em Probabilistic and Statistical Aspects of Quantum Theory}, North-Holland 1982 (reissued by Edizoni della Normale, 2011)

\bibitem{Unsharp} J. Harding and A. Wilce, Probabilistic theories and order-unit spaces, to appear

\bibitem{Johnson-Yau} N. Johnson and D. Yau, 
{\em $2$-Dimensional Categories}, 
{\tt https://arxiv.org/pdf/2002.06055} (2020)


\bibitem{Leifer} M. Leifer, Is the quantum state real? An extended review of $\psi$-ontology theorems, Quanta {\bf 3} 2014;	

\bibitem{Maudlin} T. Maudlin, What Bell did, J. Phys. A {\bf 47}, 2014 {\tt arXiv:1408.1826} 

\bibitem{PR} S. Popescu and D. Rohrlich, Quantum nonlocality as an axiom, Foundations of Physics {\bf 24} (1994)


\bibitem{Simon} B. Simon, {\em Convexity: an analytic viewpoint},  Cambridge, 2011

\bibitem{Spekkens} R. Spekkens, Contextuality for preparations, transformations, and unsharp measurements, Phys. Rev. A {\bf 71}  (2005)

\bibitem{Werner} R. Werner, Comment on Maudlin's Paper ``What Bell did", J. Phys. A {\bf 47}, 2014; see also R. Werner, 
What Maudlin replied to, 
{\tt https://arxiv.org/abs/1411.2120} (2014)

\bibitem{SEP} A. Wilce, Quantum logic and probability theory, Stanford Encyclopedia of Philosophy 2002

\bibitem{Tesp} A. Wilce, Test spaces, in D. Gabbay, K. Engesser and D. Lehman (eds.) {\em Handbook of Quantum Logic: Quantum Structures}, North-Holland, 2006 [check]

\bibitem{Shortcut} A. Wilce, A shortcut from categorical quantum theory to convex operational theories, in Bob Coecke and Aleks Kissinger (eds.): 14th International Conference on Quantum Physics and Logic (QPL) EPTCS {\bf 266}, 2018, pp. 222–236


\bibitem{CCAM} Coarse-graining and compounding as monads, {\tt 
arXiv:2410.08818} (2024)

\bibitem{Short-Course} A. Wilce, {\em Generalized Probability Theory: notes for a short course} {\tt arXiv:2501.00718} (2025)

\bibitem{Wright} R. Wright, The structure of projection-valued states: A generalization of Wigner's theorem, 
Int. J. Theor. Physics {\bf 16W}, 1977


\end{thebibliography}
\end{document}